\documentclass[letterpaper, 10 pt, conference]{ieeeconf}  

\IEEEoverridecommandlockouts                              
\usepackage{graphics} 
\usepackage{epsfig} 
\usepackage{amsmath} 
\usepackage{amssymb}  

\usepackage{tikz}
\usepackage{eso-pic}
\usepackage{xcolor}

\makeatletter
\let\NAT@parse\undefined
\makeatother
\usepackage{hyperref}
\hypersetup{
	colorlinks=false,
	linkcolor=black,
	urlcolor=black,
	hidelinks
}

\renewcommand\fbox{\fcolorbox{black}{gray!10}}

\newcommand\acceptedmanuscripttext{%
	\scriptsize
	\textbf{\textcopyright 2025 IEEE.} Personal use of this material is permitted. Permission from IEEE must be obtained for all other uses, in any current or future media, including reprinting/republishing this material for advertising or promotional purposes, creating new collective works, for resale or redistribution to servers or lists, or reuse of any copyrighted component of this work in other works.\\

	\textbf{This manuscript was presented at the \textit{2025 IEEE 19th International Conference on Control \& Automation (ICCA)}}.
	The final Version of Record is available at \url{https://doi.org/10.1109/ICCA65672.2025.11129781}.
	Please cite the published version.
}

\AddToShipoutPicture*{%
	\AtPageUpperLeft{%
		\begin{tikzpicture}[remember picture,overlay]
			\node[
			anchor=north,
			yshift=-7pt
			] at (current page.north)
			{%
				\fbox{%
					\parbox{\dimexpr\textwidth-2\fboxsep-2\fboxrule}{%
						\acceptedmanuscripttext
					}%
				}%
			};
		\end{tikzpicture}
	}%
}

\title{\LARGE \bf
	Orbital Station-Keeping in the Earth-Moon System via~Nonlinear~Backstepping
}

\author{António Nunes, Pedro Batista, and Sérgio Brás
	\thanks{This work was supported by LARSyS FCT funding (DOI: 10.54499/LA/P/0083/2020, 10.54499/UIDP/50009/2020, and 10.54499/UIDB/50009/2020), and by NEURASPACE Project, Contract No. 9, under Regulation (EU) 2021/241 of the European Parliament and of the Council of February 12, 2021 and the Portuguese Recovery and Resilience Program (PRR), in component 05-Capitalization and Business Innovation, under Notice No. 425 01/C05-i01/2021 of the Regulation of Mobilizing Agendas/Alliances for reindustrialization}
	\thanks{António Nunes is with the Institute for Systems and Robotics, Instituto Superior Técnico, Universidade de Lisboa, Portugal (e-mail: {\tt \footnotesize \href{mailto:antonio.w.nunes@tecnico.ulisboa.pt}{antonio.w.nunes@tecnico.ulisboa.pt}}) }
	\thanks{Pedro Batista is with the Institute for Systems and Robotics,
		Instituto Superior Técnico, Universidade de Lisboa, Portugal}%
	\thanks{Sérgio Brás is with the European Space Agency, The Netherlands}%
}

\newtheorem{corollary}{Corollary}[section]

\newtheorem{proposition}{Proposition}[section]

\newcommand{\norm}[1]{|\!|#1|\!|}

\begin{document}

\maketitle
\thispagestyle{empty}
\pagestyle{empty}

\begin{abstract}
	
A nonlinear orbital station-keeping solution for the circular and elliptic versions of the Earth-Moon Restricted Three-Body Problem (R3BP) is developed via a backstepping technique. Formal guarantees for global asymptotic stability (GAS) are attained, as shown through Lyapunov's stability theory. The adequacy of the proposed control law is evaluated through the means of numerical trials over closed periodic solutions of the circular and elliptic R3BPs. The ramifications of the control gain choice are carefully studied and simulated.
	
\end{abstract}

\section{Introduction}
	
Recently, there has been a significant commitment to returning mankind to the lunar surface, following the efforts of NASA's Artemis program and other initiatives alike. To this end, the Restricted Three-Body Problem (R3BP) is known to provide an adequate first approximation to the dynamics of a spacecraft in the Earth-Moon system. In fact, closed periodic solutions around its five Lagrange (equilibrium) points have been identified as ideal trajectories for contemporary missions within this scope, due to their favorable conditions in terms of positioning with respect to the Moon and communication clearance with the Earth. Despite their appeal, these orbits are typically unstable, meaning that control action is necessary if a spacecraft is to track such a trajectory. These control solutions are usually developed with the intent of driving deviations from a target point evolving along the nominal trajectory to zero, which formally constitutes an orbital \textit{station-keeping} problem.

In the literature, this problem has been addressed primarily through a linearization of the equations of motion (EoM) about a target equilibrium solution. The linear approach to the orbital station-keeping problem has historically been preferred since it allows for a juxtaposition with well-founded fields within linear theory. Renowned examples include the use of concepts from Floquet theory on linear periodic systems to develop feedback control of the unstable (linear) orbital modes \cite{wiesel1983ModalControl}, and invariant manifold theory \cite{simo1986StationkeepingInvManif}, for a \textit{looser} control. The ease of including a measure of optimality in the design of linear control laws has also contributed to their preference over other alternatives. Examples range from more classical \textit{target-point} approaches \cite{howell1993StationkeepingPointMethod}, where the projected spacecraft deviation is weighted against current correction maneuvers, to the more modern implementations of Model Predictive Control (MPC) found in \cite{Elobaid2022MPC} and \cite{CuevasDelValle2022ModernControl}.

Despite the benefits outlined, the absence of formal guarantees for global asymptotic stability (GAS) under the full, more realistic, nonlinear (NL) model of the R3BP dynamics is a clear disadvantage of linear strategies. As such, some work has been devoted to the development of alternative NL solutions. Recent approaches of relevance include the use of output regulation theory \cite{akiyama2018outputRegulation} or sliding mode control \cite{darvish2015slidingMode} for station-keeping and formation flying. However, there is a gap in the literature for a NL approach that is easily adaptable to dynamical models of varying fidelity, while maintaining a computationally light profile -- which is especially relevant in space applications, where numerical power is limited.

This work proposes a novel NL control law that is developed via a standard and lightweight backstepping technique.  For the sake of versatility, its design and application are covered for the two most common dynamical models for spacecraft motion in the Earth-Moon system: the circular and elliptic R3BPs. The proposed solution is shown to guarantee GAS when targeting any orbit within both problems, while keeping computational demands at bay.

The remainder of this paper is organized as follows. Firstly, the dynamics of the two R3BP formulations are introduced in Section~\ref{sec:dyn}, highlighting key differences and equilibrium (periodic) solutions of interest. Then, the controller design is outlined in Section~\ref{sec:cont}, where GAS is formally shown. The proposed solution is thereafter evaluated numerically in Section~\ref{sec:results} and its adequacy is explored for various choices of control gains. Finally, the most meaningful conclusions and possible future work are discussed in Section~\ref{sec:conc}.

\section{Dynamics}
\label{sec:dyn}

In its most general definition, the R3BP models the dynamics of a small mass (e.g., a spacecraft), that feels the gravitational pull of two much more massive bodies in its proximity, named primaries, but has a negligible contribution to their motion. In this regard, the trajectories of the primaries are known to lie on a limited set of possibilities \cite{szebehely1967TheoryOfOrbits}. Of particular interest is the case of Keplerian orbits around the shared barycenter, which provides an adequate approximation for the Earth-Moon relative motion. Within this scope, two scenarios may be considered:
(i) the simpler case, where the orbits of the primaries are further approximated to being circular or (ii) the more general case, where elliptical orbits are considered instead, providing a more realistic, albeit more complex, representation of the primaries' motion. When the third body is introduced, these two cases give rise to the Circular Restricted Three-Body Problem (CR3BP) and Elliptic Restricted Three-Body Problem (ER3BP), respectively.

In both versions of the R3BP, the dynamics of the third body are established under a \textit{synodic} reference frame that aligns both primaries along $X$ and rotates to match their angular speed, directed along $Z$. The EoM are scaled to bring the distance between primaries to unity and their orbital period to $2\pi$, resulting in unit mean angular speed. In addition, their masses are scaled to $m_1=1-\mu$ and $m_2=\mu$ (considering $m_1\geq m_2$), with $\mu$ a problem-specific constant -- approximately $\mu\approx0.01215$, for the Earth-Moon system. This procedure fixes the positions of the primaries with mass $m_1$ and $m_2$ to $(-\mu,0,0)$ and $(1-\mu,0,0)$, respectively.
	
%

\subsection{Circular Restricted Three-Body Problem}

	As explored in \cite{wiesel2010SpaceflightDynamics}, under the proposed scaling and reference frame, with coordinates $(x,y,z)$, an effective potential in the CR3BP may be established as $U = \frac{x^2+y^2}{2} + \frac{1-\mu}{r_1} + \frac{\mu}{r_2}$, where $r_1$ and $r_2$ are the scalar distances between the third body and $m_1$ and $m_2$, respectively. Adopting the dot notation to represent time derivatives, and assuming no sensor noise nor external perturbations, the resulting third body EoM are
\begin{equation}
	\label{eq:eom_matrix_form}
	\begin{aligned}
		\dot{\mathbf{r}} &= \mathbf{I_3}\mathbf{v},\\
		\dot{\mathbf{v}} &= \boldsymbol{\Omega}\mathbf{v} + \nabla U(\mathbf{r}),
	\end{aligned}
	\qquad \text{with} \qquad 
	\boldsymbol{\Omega}=\begin{bmatrix}
		0 & 2 & 0 \\ -2 & 0 & 0 \\ 0 & 0 & 0
	\end{bmatrix},
\end{equation}
where ${\mathbf{r}^T=\begin{bmatrix}
		x & y & z
\end{bmatrix}}$ is its position, ${\mathbf{v}^T=\begin{bmatrix}
\dot{x} & \dot{y} & \dot{z}
\end{bmatrix}}$ is its velocity, $\mathbf{I_3}$ is the $3\times3$ identity matrix, and $\nabla U(\mathbf{r})$ is the gradient of the effective potential, evaluated at $\mathbf{r}$.

\subsection{Elliptic Restricted Three-Body Problem}
\label{subsec:er3bp}
	
In the ER3BP, it is important to highlight that the angular velocity of the primaries and the distance between them are not constant. Ultimately, this means that the synodic reference frame rotates at non-constant angular speed and pulsates to keep the coordinates of the primaries fixed. Furthermore, it is common to adopt the \textit{true anomaly}\footnote{The true anomaly is the angular position, in radians, of the primaries along their orbits, with respect to the major axis. The primaries are at their closest at \textit{periapsis}, for $\nu=0$, and at their farthest at \textit{apoapsis}, for $\nu=\pi$.} of the primaries, $\nu$, as the independent variable of the problem, instead of time \cite{szebehely1967TheoryOfOrbits}. In this sense, representing quantities expressed in the synodic, pulsating reference frame of the ER3BP by $\bar{(~\cdot~)}$, an equivalent effective potential may be established as $\bar U = \frac{r}{p} \left( \frac{\bar x^2 + \bar y^2}{2} + \frac{1-\mu}{\bar r_1} + \frac{\mu}{\bar r_2} - \frac{1}{2} e \cos(\nu) \bar z^2 \right)$, where ${r=\frac{p}{1+e\cos(\nu)}}$ is the dimensionless distance between primaries, $e$ is the eccentricity of their orbits, and ${p := (1-e^2)}$. Adopting $(~\cdot~)'$ to denote derivatives with respect to $\nu$, the resulting EoM of the third body are thus
\begin{equation}
	\label{eq:eom_matrix_form_er3bp}
	\begin{aligned}
		\bar{\mathbf{r}}' &= \mathbf{I_3}\bar{\mathbf{v}},\\
		\bar{\mathbf{v}}' &= \boldsymbol{\Omega}\bar{\mathbf{v}} + \nabla \bar{U}(\bar{\mathbf{r}},\nu),
	\end{aligned}
	\quad \text{with} \qquad
	\begin{aligned}
		\bar{\mathbf{r}}^T &=\begin{bmatrix}
			\bar{x} & \bar{y} & \bar{z}
		\end{bmatrix},\\
		\bar{\mathbf{v}}^T &=\begin{bmatrix}
			\bar{x}' & \bar{y}' & \bar{z}'
		\end{bmatrix}.
	\end{aligned}
\end{equation}
which, in contrast with \eqref{eq:eom_matrix_form}, are non-autonomous.

\subsection{Lagrange Points and Periodic Orbits}

Despite differing in complexity, EoM \eqref{eq:eom_matrix_form} and \eqref{eq:eom_matrix_form_er3bp} admit the same five equilibrium points, so-called Lagrange points, which are ordered $L_1$ through $L_5$ by increasing potential \cite{szebehely1967TheoryOfOrbits}. These lie on the orbital plane of the primaries ($X$--$Y$) such that $L_1$, $L_2$, and $L_3$ are colinear with the $X$ direction, while $L_4$ and $L_5$ form equilateral triangles with the two bodies.

While some missions target the Lagrange points directly, the vast majority consider closed periodic orbits around them, given their variety and versatility. In fact, Poincaré's conjecture \cite{poincare1892MethodesNouvelles} shows the existence of an infinite amount of closed periodic solutions in the CR3BP, which is not reflected in the ER3BP, given the non-autonomous nature of its EoM.

\section{Controller Design}
\label{sec:cont}
To design the control solutions, actuation is considered in terms of accelerations, where independent action over the three base directions is assumed. The exposition of its design for both R3BP dynamical models is here briefly approached.

\subsection{Nonlinear Control Law in the CR3BP}
To exploit the integral nature of the EoM \eqref{eq:eom_matrix_form} and \eqref{eq:eom_matrix_form_er3bp}, a backstepping procedure is employed in the design of a NL control law. Identifying quantities along a nominal orbit with $(~\cdot~)^*$, control is introduced for a deviation-based version of the dynamics. For the CR3BP, this corresponds to
\begin{equation}
	\label{eq:z_1_z_2_dynamics}
		\dot{\mathbf{z}}_1 = \mathbf{z}_2 \qquad \text{and} \qquad
		\dot{\mathbf{z}}_2  = \mathbf{f}_a(\mathbf{z}_1,\mathbf{z}_2) + \mathbf{u},
\end{equation}
where $\mathbf{z}_1=\mathbf{r}-\mathbf{r}^*$ and $\mathbf{z}_2=\mathbf{v}-\mathbf{v}^*$ are the position and velocity deviations from the nominal orbit, respectively, $\mathbf{f}_a(\mathbf{z}_1,\mathbf{z}_2)  =  \boldsymbol{\Omega}\mathbf{z}_2  + \nabla U(\mathbf{z}_1+\mathbf{r}^*) -\nabla U(\mathbf{r}^*)$ is the deviation-induced acceleration error, and $\mathbf{u}$ is the control variable.

In light of a typical backstepping procedure, $\mathbf{z}_2$ is interpreted as a \textit{virtual} input to the dynamics of $\mathbf{z}_1$, in \eqref{eq:z_1_z_2_dynamics}. To this end, the state-feedback law $\mathbf{z}_2=\boldsymbol{\phi}(\mathbf{z}_1)=-\mathbf{K}_1\mathbf{z}_1$, where $\mathbf{K}_1\succ 0$, is pursued, and the Lyapunov candidate function
\begin{equation*}
	V(\mathbf{z}_1) = \frac{1}{2}\mathbf{z}_1^T\mathbf{z}_1,~ \text{with}~ V(\mathbf{0})=0~ \text{and}~ V(\mathbf{z}_1)>0,\forall \mathbf{z}_1\neq \mathbf{0},
\end{equation*}
is defined. Given that $V(\mathbf{z}_1)$ is radially unbounded and that
\begin{equation*}
	\dot{V}(\mathbf{z}_1)=\nabla V^T(\mathbf{z}_1)\boldsymbol{\phi}(\mathbf{z}_1) = -\mathbf{z}_1^T\mathbf{K}_1\mathbf{z}_1,
\end{equation*}
satisfies $\dot{V}(\mathbf{z}_1)<0,~\forall \mathbf{z}_1\neq \mathbf{0}$, the state-feedback law $\boldsymbol{\phi}(\mathbf{z}_1)$ ensures GAS of the origin of the $\mathbf{z}_1$ dynamics, following the Barbashin-Krasovskii theorem \cite[Theorem~4.2]{khalil2001NLSystems}, an extension to Lyapunov's direct method  \cite[Theorem~4.1]{khalil2001NLSystems}.

Through the change of coordinates $ \mathbf{z}_2 = \boldsymbol{\zeta} +\boldsymbol{\phi}(\mathbf{z}_1)$, the differential equations \eqref{eq:z_1_z_2_dynamics} lead to the cascade connection
\begin{equation*}
	\dot{\mathbf{z}}_1 = \boldsymbol{\zeta} + \boldsymbol{\phi}(\mathbf{z}_1) \qquad \text{and} \qquad \dot{\boldsymbol{\zeta}} = \boldsymbol{\upsilon},
\end{equation*}
where $\boldsymbol{\upsilon} := \mathbf{u} + \mathbf{f}_a(\mathbf{z}_1,\mathbf{z}_2) - \dot{\boldsymbol{\phi}}(\mathbf{z}_1)$. These dynamics are similar to the original form, with the difference being that the origin of the first component is asymptotically stable when ${\boldsymbol{\zeta}=\mathbf{0}}$.  This is exploited in the design of $\boldsymbol{\upsilon}$ to stabilize the entire system, namely by proposing the composite Lyapunov candidate function $V_c(\mathbf{z}_1,\boldsymbol{\zeta})=V(\mathbf{z}_1)+\frac{1}{2}\boldsymbol{\zeta}^T\boldsymbol{\zeta}$. Its derivative, \begin{equation*}
	\dot{V}_c(\mathbf{z}_1,\boldsymbol{\zeta}) = \nabla{V}^T(\mathbf{z}_1)\left[\boldsymbol{\zeta} + \boldsymbol{\phi}(\mathbf{z}_1)\right] + \boldsymbol{\zeta}^T\boldsymbol{\upsilon} ,
\end{equation*}
hints towards the control law ${\boldsymbol{\upsilon} = -\nabla V(\mathbf{z}_1) - \mathbf{K}_2\boldsymbol{\zeta}}$, with $\mathbf{K}_2\succ 0$. Thus, by reverting the change of coordinates and expanding, one gets to the following result in terms of $\mathbf{u}$.

\begin{proposition}
	\label{prop:control_law_NL_CR3BP}
	The  NL control law
	\begin{equation}
		\label{eq:NL_backstepping_controller_law}
		\mathbf{u} = -(\mathbf{I_3}+\mathbf{K}_2\mathbf{K}_1)\mathbf{z}_1-(\mathbf{K}_1+\mathbf{K}_2)\mathbf{z}_2 - \mathbf{f}_a(\mathbf{z}_1,\mathbf{z}_2),
	\end{equation}
	with gain matrices $\mathbf{K}_1\succ0$ and $\mathbf{K}_2\succ0$, guarantees GAS of the origin of the CR3BP EoM \eqref{eq:z_1_z_2_dynamics}, in the sense of Lyapunov.
\end{proposition}
\begin{proof}
	Consider the Lyapunov candidate function
	\begin{equation}
		\label{eq:lyap_func}
		V_c(\mathbf{z}_1,\mathbf{z}_2) = \frac{1}{2}\mathbf{z}_1^T\mathbf{z}_1 + \frac{1}{2}(\mathbf{z}_2+\mathbf{K}_1\mathbf{z}_1)^T(\mathbf{z}_2+\mathbf{K}_1\mathbf{z}_1),
	\end{equation}
	which clearly satisfies $V_c(\mathbf{0},\mathbf{0})=0$ and $V_c(\mathbf{z}_1,\mathbf{z}_2)>0,$ ${\forall(\mathbf{z}_1,\mathbf{z}_2)\neq(\mathbf{0},\mathbf{0})}$, for $\mathbf{K}_1\succ0$.
	Then, differentiating and substituting $\mathbf{u}$ from \eqref{eq:NL_backstepping_controller_law} results in
	\begin{equation*}
		\begin{aligned}
			\dot{V}_c(\mathbf{z}_1,\mathbf{z}_2) &= -\mathbf{z}_1^T\mathbf{z}_1 - (\mathbf{z}_2+\mathbf{K}_1\mathbf{z}_1)^T\mathbf{K}_2(\mathbf{z}_2+\mathbf{K}_1\mathbf{z}_1),
		\end{aligned}
	\end{equation*}
	which satisfies ${\dot{V}_c(\mathbf{z}_1,\mathbf{z}_2)<0,}$ ${\forall(\mathbf{z}_1,\mathbf{z}_2)\neq(\mathbf{0},\mathbf{0})}$, for ${\mathbf{K}_2\succ 0}$, confirming that $V_c(\mathbf{z}_1,\mathbf{z}_2)$ is a Lyapunov function. Moreover, and since \eqref{eq:lyap_func} is radially unbounded,  \eqref{eq:NL_backstepping_controller_law} ensures GAS of the origin of \eqref{eq:z_1_z_2_dynamics}, following \cite[Theorem~4.2]{khalil2001NLSystems}.
\end{proof}

\subsection{Nonlinear Control Law in the ER3BP}
In parallel to the solution proposed for the CR3BP, a control law in the ER3BP is established by pursuing very similar steps. The result is summarized as follows, using equivalent notations to denote a nominal orbit and the corresponding spacecraft deviations in position and "velocity" -- the latter of which, in this case, are derivatives with respect to $\nu$.
\begin{corollary}
	The NL control law
	\begin{equation}
		\label{eq:NL_backstepping_controller_law_ER3BP}
		\mathbf{u} = -(\mathbf{I_3}+\mathbf{K}_2\mathbf{K}_1)\bar{\mathbf{z}}_1-(\mathbf{K}_1+\mathbf{K}_2)\bar{\mathbf{z}}_2 - \bar{\mathbf{f}}_a(\bar{\mathbf{z}}_1,\bar{\mathbf{z}}_2, \nu),
	\end{equation}
	with gain matrices $\mathbf{K}_1 \succ 0$ and $\mathbf{K}_2 \succ 0$, ensures (uniform) GAS of the origin of the ER3BP dynamics, given by
	\begin{equation*}
		\bar{\mathbf{z}}'_1 = \bar{\mathbf{z}}_2 \qquad \text{and} \qquad
		\bar{\mathbf{z}}'_2  = \bar{\mathbf{f}}_a(\bar{\mathbf{z}}_1,\bar{\mathbf{z}}_2,\nu) + \mathbf{u},
	\end{equation*}
	where $\bar{\mathbf{f}}_a(\bar{\mathbf{z}}_1,\bar{\mathbf{z}}_2, \nu)=\boldsymbol{\Omega}\bar{\mathbf{z}}_2  + \nabla \bar{U}(\bar{\mathbf{z}}_1+\bar{\mathbf{r}}^*,\nu) -\nabla \bar{U}(\bar{\mathbf{r}}^*,\nu)$.
\end{corollary}
\begin{proof}
	The proof is identical to that of Proposition~\ref{prop:control_law_NL_CR3BP}, since the equivalent Lyapunov function and its derivative are independent of $\nu$ -- a trivial case of Lyapunov's analogous theorem for non-autonomous systems \cite[Theorem~4.9]{khalil2001NLSystems}.
\end{proof}

Note that control law \eqref{eq:NL_backstepping_controller_law_ER3BP} depends explicitly on the respective independent variable, $\nu$, whereas \eqref{eq:NL_backstepping_controller_law} does not, reflecting the non-autonomous nature of the ER3BP. Whether this apparent increase in complexity is of concern to the controller operation is evaluated via numerical trials in Section~\ref{sec:results}.

\subsection{Remarks}
\label{subsec:obs}

In principle, any positive definite gain matrices $\mathbf{K}_1$ and $\mathbf{K}_2$ are equally valid for the assurance of GAS. However, to better study how these matrices may be tailored to fulfill particular mission requirements, this work considers the simpler case where they are multiples of the identity matrix, i.e. $\mathbf{K}_1=k_1\mathbf{I_3}$ and $\mathbf{K}_2=k_2\mathbf{I_3}$, with $k_1>0$ and $k_2>0$.

To predict how $k_1$ and $k_2$ may impact the behavior of the proposed control laws, it is pertinent to separate their linear and NL terms. In the case of \eqref{eq:NL_backstepping_controller_law}, for example, one may write
\begin{equation*}
	\mathbf{u} = -\mathbf{u}_{\mathrm{lin}}-\mathbf{f}_a, \qquad \text{with} \qquad  \mathbf{u}_{\mathrm{lin}} = (1+k_1 k_2)\left(\mathbf{z}_1 + \kappa \mathbf{z}_2\right),
\end{equation*}
where $\kappa = \frac{k_1 + k_2}{1+ k_1 k_2}$.
Doing so clearly evidences that the values of $k_1$ and $k_2$ may be freely interchanged and still result in the exact same control command, therefore leading to an equal system response.
Furthermore, by studying the behavior of $\kappa$, it is possible to assess how the gains impact the relative weight of the position and velocity deviations towards $\mathbf{u}_\mathrm{lin}$ (and thus $\mathbf{u}$).
To this end, taking both gains to be either very large or very small predominantly penalizes position deviations over those in velocity, since ${k_1\to0,~k_2\to 0 \implies \kappa \to 0}$ and ${k_1\to\infty,~k_2\to \infty \implies \kappa \to 0}$. In contrast, if the magnitude of the two gains differs significantly, the result is a heavier penalization of deviations in velocity, given that ${k_1\to 0,~k_2\to \infty \implies \kappa \to \infty}$. From a practical point of view, however, opting towards this choice may prove naive, given how the velocity deviation states correspond to direct derivatives of the position deviations, thus naturally approaching zero when the spacecraft is spatially driven towards the nominal orbit. This means that a heavier penalization of position deviations may actually bring those in velocity faster to zero -- an hypothesis later confirmed numerically in Section~\ref{sec:results}. Notice however that, while selecting the control gains to adjust how the position and velocity deviations are weighed against each other is possible, the overall control command \textit{intensity}, i.e.  in terms of its magnitude, cannot be freely tuned simultaneously.

\section{Numerical Results}
\label{sec:results}

To evaluate the proposed solution, a target orbit in the CR3BP (C1) and another in the ER3BP (E1) are selected. C1 was retrieved from the extensive catalog in \cite{jpl3BPorbits} and belongs to the L2 Halo family, often selected for near-Moon missions. C2 is an equivalent double revolution L2 orbit derived iteratively from C1, following the steps in \cite{ferrari2018ER3BP}.
These may be obtained through the numerical integration of the appropriate dynamics, considering the dimensionless initial conditions 
${\mathbf{x}_0^T\approx\begin{bmatrix}
	1.1438 & 0 & -0.1575 & 0 & -0.2219 & 0
\end{bmatrix}}$
and
${\bar{\mathbf{x}}_0^T\approx\begin{bmatrix}
	1.1452 & 0 & -0.1609 & 0 & -0.2209 & 0
\end{bmatrix}}$, respectively.
The orbital period and Poincaré exponents of each orbit are given in Table~\ref{tab:test_cases_stab}, where $j$ denotes the imaginary unit. Note the existence of a single unstable mode in both orbits, whose magnitude reveals high instability, making these trajectories ideal candidates for evaluating the adequacy of the proposed solution under extreme-case scenarios.

\begin{table}[h]
	\renewcommand*{\arraystretch}{1.2}
	\caption{Stability properties and period of the nominal orbits.}
	\label{tab:test_cases_stab}
	\begin{center}
		\begin{tabular}{c c c}
			\hline
			Orbit& Poincaré Exponents & Orbital Period\\ \hline
			C1 & $(\pm 1.607,~ \pm 0.572j,~ 0,~ 0)$ & $\sim13.9$ days\\
			E1 & $(\pm 1.609,~ \pm 0.430j,~ \pm 0.004j)$ & $\sim27.9$ days\\ \hline
		\end{tabular}
	\end{center}
\end{table}

The test-case under analysis considers an initial deviation ${\mathbf{z}_0^T = [(300,-300,300) ~\text{km},~(-0.5,0.5,-0.5)~\text{m s}^{-1}]}$ from each orbit, which the controller aims to eliminate. The dynamics are simulated by numerical integration of the EoM through a 9$^{\text{th}}$ order adaptive-step Runge-Kutta method, with absolute and relative error tolerances set to $10^{-12}$. Note however that the control law is independent to this process. In fact, since all of its terms may be evaluated analytically, the proposed strategy maintains a computationally light profile.

To quantitatively evaluate the controller performance, actuation effort benchmarks are defined over the first orbit as
\begin{equation*}
	\mathcal{E}_{v}= \int_{0}^{T} \norm{\mathbf{u}(t)} dt \qquad \text{and} \qquad \mathcal{E}_{e} = \int_{0}^{T} \norm{\mathbf{u}(t)}^2 dt,
\end{equation*}
where $T$ is the corresponding orbital period. 
These provide measures for the velocity variation enacted by the actuators during that time and the corresponding energy expenditure, respectively. In addition, to assess the timescale of the responses, let $t_{m}$ denote the time instant when the position deviation is within meter-level range, i.e. $\norm{\mathbf{z}_1}<10~\text{m}$.

\subsection{Application in the CR3BP}
\label{subsec:app_CR3BP}

The results stemming from the application of the NL control law \eqref{eq:NL_backstepping_controller_law} in the CR3BP are presented in Fig.~\ref{fig:CR3BP_resp}. These are given in terms of the magnitudes of the position deviation, velocity deviation, and control command, over the course of the first 7 days of operation. For the sake of clarity, these quantities are scaled back to physical units by reverting the nondimensionalization steps outlined in Section~\ref{sec:dyn}.To evaluate the effect of the control gains, each plot provides the responses under four different choices of $(k_1,k_2)$. A detailed view of the initial control commands is provided to evidence the subtle differences between responses. Recall from the discussion in Section~\ref{subsec:obs} that $k_1$ and $k_2$ may be freely interchanged without consequence, hence why it is most relevant to assess them from a relative perspective.

The analysis of Fig.~\ref{fig:CR3BP_resp} confirms the hypotheses postulated in Section~\ref{subsec:obs}. On the one hand, it shows that for two control gains of similar magnitude, e.g. $(k_1,k_2)=(10,10)$, the position deviations are brought faster to zero. In contrast, if the gains differ significantly, e.g. $(k_1,k_2)=(0.1,5)$, the velocity deviation is most contained, given the larger contribution of this variable towards the control command. However, as predicted, this does not imply a faster convergence of $\mathbf{z}_2$ towards zero. In fact, while choosing $k_1=k_2$ predominantly penalizes deviations in position, which results in a considerable initial overshoot in $\mathbf{z}_2$ (especially for large gains), note that its magnitude under such a gain selection eventually goes below that of the opposing alternatives. As explored, this is a direct consequence of the integral relation between position and velocity, i.e. that the elimination of $\mathbf{z}_1$ results in eliminating also $\mathbf{z}_2$. On the contrary, the test-cases where deviations in velocity are more heavily penalized show a slow convergence of $\mathbf{z}_1$ towards zero, which is reflected in an incomplete elimination of $\mathbf{z}_2$ -- an effect entirely opposite to what may have been initially conceptualized.

To further underline the objective consequences of the control gain selection, Table~\ref{tab:efforts_CR3BP} provides the total effort benchmarks and timescales of the responses in Fig.~\ref{fig:CR3BP_resp}, as defined in Section~\ref{sec:results}. Once more, the results highlight the relevance of prioritizing the elimination of deviations in position. To illustrate this point, notice how the large timescale of the response with ${(k_1,k_2)=(0.1,5)}$ is ultimately responsible for a larger value of the control effort benchmark measuring the total velocity variation, $\mathcal{E}_v$, despite the relatively small command magnitude throughout. The analysis of the total control energy expenditure, $\mathcal{E}_{e}$, leads to a similar conclusion. In addition, it is possible to see an increase in both benchmarks as the control gains are made larger, due to the more vehement initial command.

\begin{figure}[t]
	\centering
	\includegraphics[width=0.98\linewidth, trim={0.3cm 0cm 1cm 0.7cm}, clip]{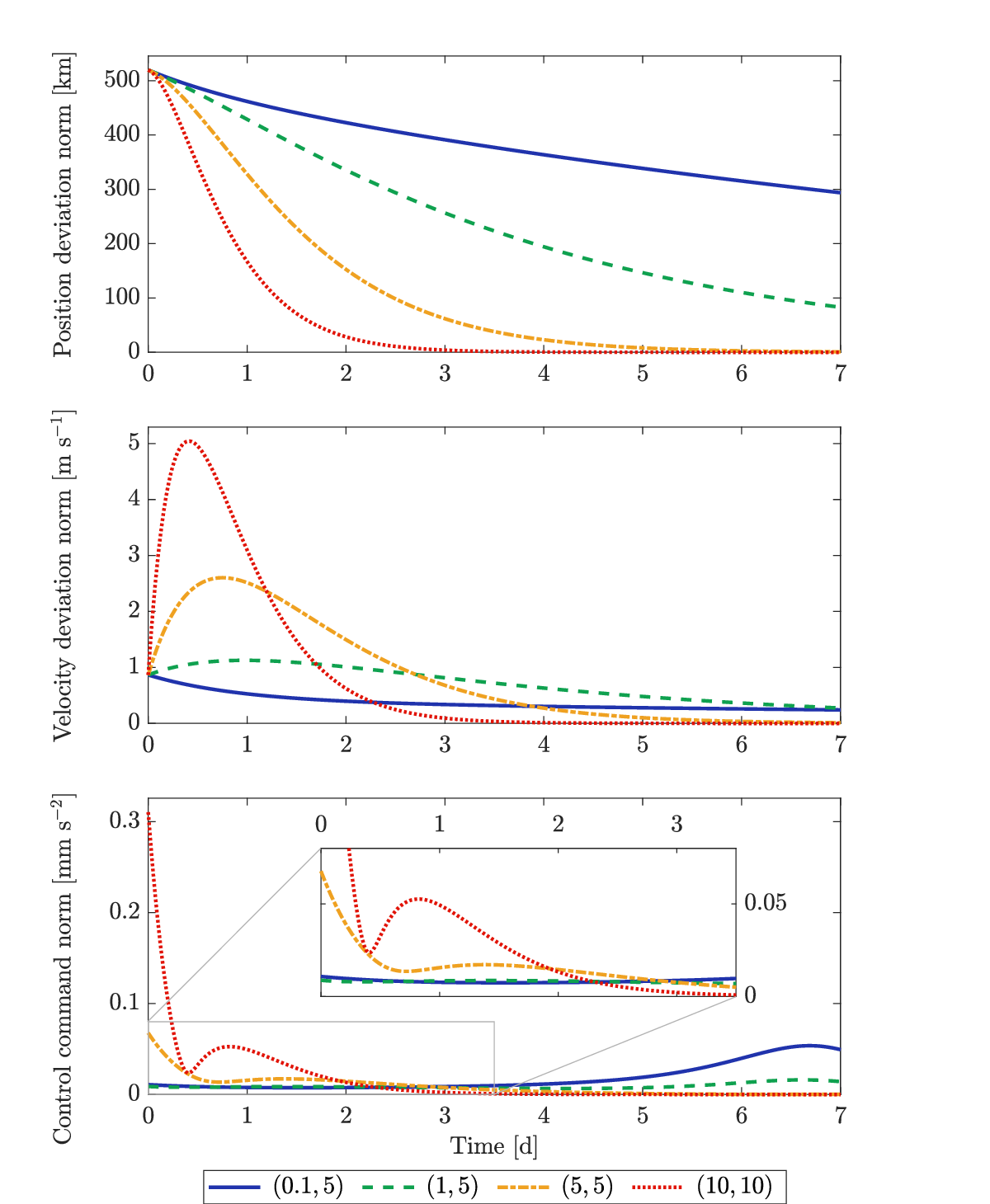}
	\caption{System responses under the CR3BP dynamics and the NL control law, considering various choices for the control gain pair $(k_1, k_2)$, over the first 7 days of simulation.}
	\label{fig:CR3BP_resp}
\end{figure}

\begin{table}[h]
	\renewcommand*{\arraystretch}{1.2}
	\caption{Control effort benchmarks over first orbital period and timescale of the responses in Fig.~\ref{fig:CR3BP_resp}.}
	\label{tab:efforts_CR3BP}
	\begin{center}
		\begin{tabular}{c c c c c c}
			\hline
			& \multicolumn{4}{c}{$(k_1,k_2)$} & \\
			\cline{2-5}
			& $(0.1,5)$ & $(1,5)$ & $(5,5)$ & $(10,10)$ &\\ \hline
			$\mathcal{E}_v$ & 15.46 & 6.732 & 5.311 & 9.788 & $(\times 10^{-3})$\\
			$\mathcal{E}_e$ & 43.74 & 6.262 &  11.59 & 96.83 & $(\times 10^{-11})$\\
			$t_m$ & $>T$ & $>T$ & 10.41 & 5.802 & (days)\\ \hline
		\end{tabular}
	\end{center}
\end{table}

Given these considerations, unless an upper-bound restriction to the velocity deviation needs to be met, there is no benefit to penalizing primarily $\mathbf{z}_2$. The options with $k_1\sim k_2$ are shown to yield more satisfactory results not only qualitatively, but also objectively, as measured through the established control effort benchmarks. If a faster response is required, both control gains may be increased in unison. This should however be carefully leveled with the sacrifice in terms of energy expenditure, as diminishing returns are observed when the control gains are increased beyond $(k_1,k_2)=(5,5)$, for the test-case presented.

Fig.~\ref{fig:CR3BP_orbit} provides the trajectory of the controlled spacecraft, considering ${(k_1,k_2)=(5,5)}$. To underline the necessity for control action, the path described by an uncontrolled spacecraft is also presented, which is seen to quickly depart from the nominal orbit due to its unstable nature. In addition, a detailed $X$--$Y$ view of the first and final instants of operation evidences the elimination of the initial deviation. The results illustrate how the proposed orbital station-keeping solution is adequate at tracking the Halo orbit considered, with the position deviation falling within meter-level range even before a full orbit is completed, in accordance with Table~\ref{tab:efforts_CR3BP}.

\begin{figure}[t]
	\centering
	\includegraphics[width=\linewidth, trim={0.5cm 0cm 0.6cm 0.9cm}, clip]{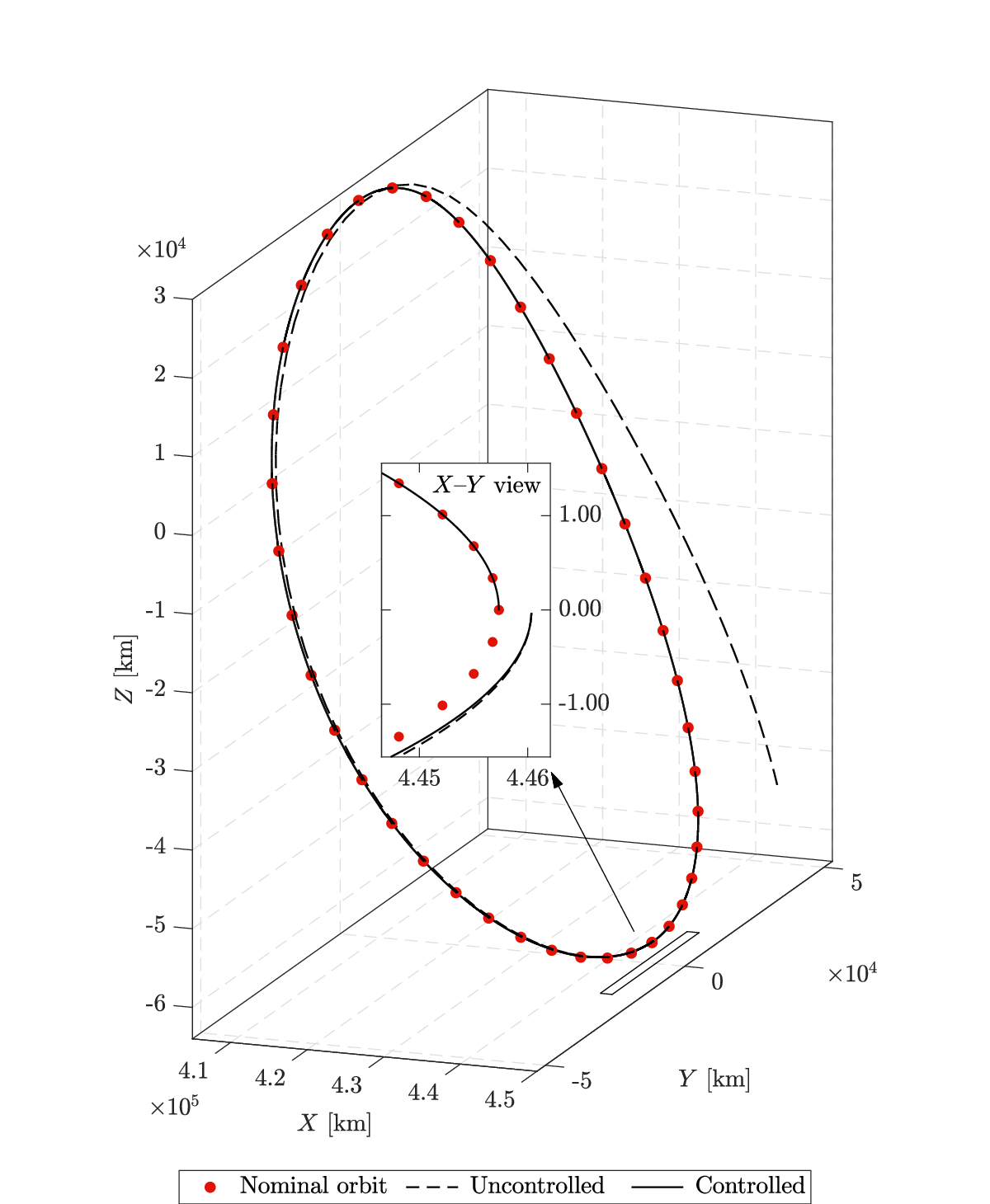}
	\caption{Orbital station-keeping in the CR3BP via the NL control law with ${(k_1,k_2)=(5,5)}$. Plotted against the trajectory of an uncontrolled spacecraft. The $X$--$Y$ view details the initial and final instants of simulation.}
	\label{fig:CR3BP_orbit}
\end{figure}

\subsection{Application in the ER3BP}

The application of the NL control law \eqref{eq:NL_backstepping_controller_law_ER3BP} under the ER3BP dynamics leads to the system responses provided in Fig.~\ref{fig:ER3BP_resp}. Once more, the magnitudes of the position deviation, velocity deviation, and actuation command are provided for the first 7 days of operation. However, the change of independent variable discussed in Section~\ref{subsec:er3bp} means further attention is required when scaling these quantities into physical units. Most notably, a time-based representation requires that "velocity" deviations are multiplied by the varying angular speed of the reference frame to ensure they represent time-derivatives. Furthermore, Kepler's equation is used to transform $\nu$ into $t$, which may be written as \cite{wiesel2010SpaceflightDynamics}
\begin{equation}
	\label{eq:kep}
	t-t_0 = E- e\sin E, \quad \text{with} \quad \tan \frac{E}{2}= \sqrt{\frac{1-e}{1+e}}\tan \frac{\nu}{2},
\end{equation}
given that the dimensionless mean angular velocity of the primaries is equal to unity.
The primaries are assumed to start at \textit{periapsis}, i.e. with ${\nu\vert_{t=0}=0}$, such that one has $t_0=0$. 

\begin{figure}[t]
	\centering
	\centering
	\includegraphics[width=0.98\linewidth, trim={0.3cm 0cm 1cm 0.7cm}, clip]{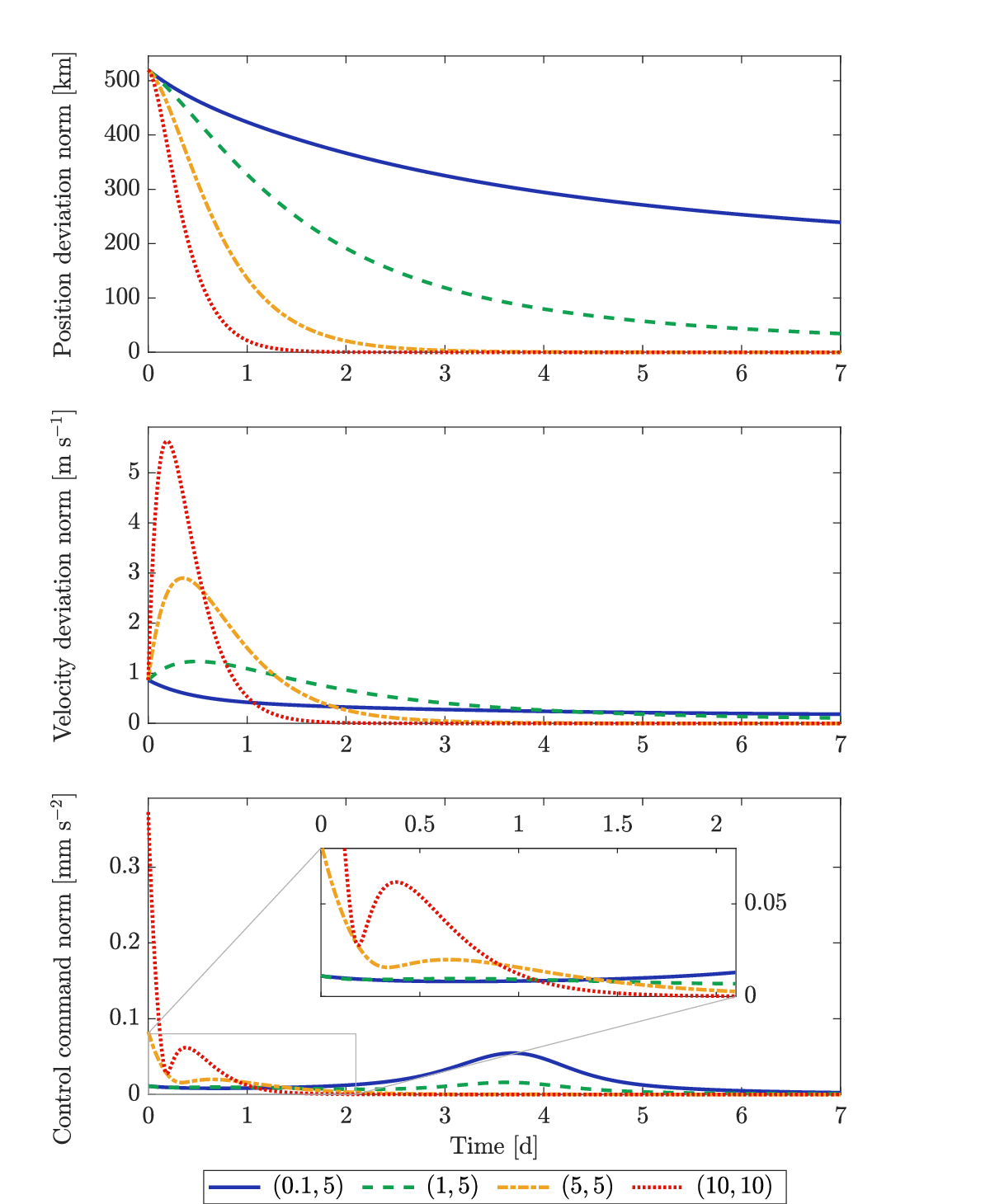}
	\caption{System responses under the ER3BP dynamics and the NL control law, considering various choices for the control gain pair $(k_1, k_2)$, over the first 7 days of simulation.}
	\label{fig:ER3BP_resp}
\end{figure}

By selecting the same control gain pairs considered during the study developed in Section~\ref{subsec:app_CR3BP}, for the CR3BP,  it is possible to directly assess if the more complex, non-autonomous dynamics of the ER3BP have any meaningful effect towards the adequacy of the proposed control solution. To this end, Table~\ref{tab:efforts_ER3BP} provides the corresponding control effort benchmarks and timescales of the responses in Fig.~\ref{fig:ER3BP_resp}.

\begin{table}[!h]
	\renewcommand*{\arraystretch}{1.2}
	\caption{Control effort benchmarks over first orbital period and timescale of the responses in Fig.~\ref{fig:ER3BP_resp}.}
	\label{tab:efforts_ER3BP}
	\begin{center}
		\begin{tabular}{c c c c c c}
			\hline
			& \multicolumn{4}{c}{$(k_1,k_2)$} & \\
			\cline{2-5}
			& $(0.1,5)$ & $(1,5)$ & $(5,5)$ & $(10,10)$ &\\ \hline
			$\mathcal{E}_v$ & 17.15 & 4.721 & 2.994 & 5.391 & $(\times 10^{-3})$\\
			$\mathcal{E}_e$ & 38.47 & 4.280 &  7.796 & 63.87 & $(\times 10^{-11})$\\
			$t_m$ & $>T$ & $>T$ & 7.395 & 3.007 & (days)\\ \hline
		\end{tabular}
	\end{center}
\end{table}

From the analysis of Fig.~\ref{fig:ER3BP_resp} and Table~\ref{tab:efforts_ER3BP} one comes to the conclusion that the overall behavior of the proposed solution under the ER3BP is analogous to the results previously found for the CR3BP. In fact, the same observations on the effect of the control gain choice may be made, as the system responses and control effort benchmarks suggest towards opting for gains of similar magnitude, which may be increased simultaneously for a faster response, albeit at an increase in energy expenditure. Furthermore, it is possible to observe that, for the same control gain pairs, the proposed solution seems to bring the deviations in the ER3BP faster to zero when compared with the CR3BP results previously presented in Table~ \ref{tab:efforts_CR3BP}, especially for larger gains. This is however a mere consequence of the non-autonomous dynamics of the ER3BP and the NL relation between time and true anomaly, since $\nu$ varies faster with time when the primaries are near periapsis, according to Kepler's equation \eqref{eq:kep}. Notice also the slightly larger overshoot in $\mathbf{z}_2$ and more demanding initial actuation command, due to the $\nu$-dependence of the control law \eqref{eq:NL_backstepping_controller_law_ER3BP} -- events that are nonetheless relatively contained. Ultimately, this analysis brings confirmation to the versatility of the backstepping methodology employed in this work.

Fig.~\ref{fig:ER3BP_orbit} provides a spatial representation of the nominal orbit tracking under the control solution with ${(k_1,k_2)=(5,5)}$, in comparison to the trajectory of an uncontrolled spacecraft. Once more, this visually confirms the adequacy of the proposed solution at tracking the nominal orbit in the ER3BP, despite its more complex double-revolution winding.

\begin{figure}[h]
	\centering
	\includegraphics[width=\linewidth, trim={0.5cm 0cm 0.6cm 0.9cm}, clip]{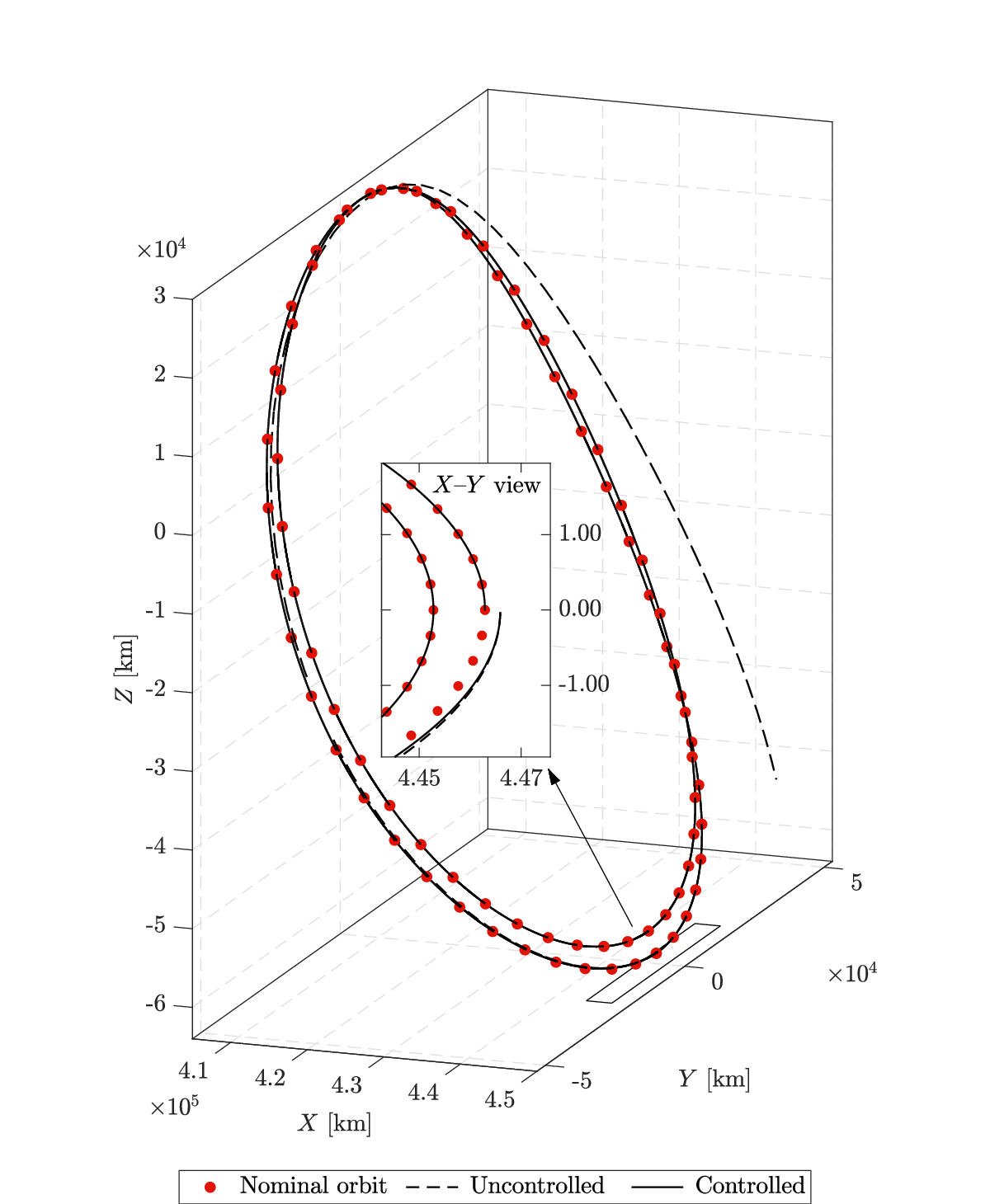}
	\caption{Orbital station-keeping in the ER3BP via the NL control law with ${(k_1,k_2)=(5,5)}$. Plotted against the trajectory of an uncontrolled spacecraft. The $X$--$Y$ view details the initial and final instants of simulation.}
	\label{fig:ER3BP_orbit}
\end{figure}
\section{Concluding Remarks}
\label{sec:conc}

This work covers the design of a NL backstepping control solution for the orbital station-keeping problem in the CR3BP and ER3BP, which presents a very light numerical profile for enhanced applicability in onboard computers. The proposed control solution is shown to provide formal guarantees for GAS under the full NL dynamics, ensuring that any deviation from a given target trajectory is asymptotically driven to zero over time. Numerical trials are carried out to evaluate the adequacy of the proposed solution and validate theoretical postulations made on the effects of the two control gains that describe it. The results suggest towards selecting control gains of similar magnitude for an ideal balance between total velocity variation, energy expenditure, and timescale of the response. They also demonstrate similar and precise tracking of closed solutions in both R3BP dynamical models studied, highlighting the adaptability of the proposed strategy, when compared to other alternatives in literature.

Future work shall assess the effect of external perturbations, e.g. due to other massive celestial bodies exerting relevant force on the spacecraft, as well as sensor/actuator noise and other error sources, which may be addressed through possible extensions via dynamic surface control or command filtered backstepping. In addition, the formal inclusion of physical constraints should be studied, namely those at the actuator-level, such as actuation saturation -- a concern that is often faced by real missions making recourse of low-propulsion electric systems to enact continuous control.

%



\bibliographystyle{IEEEtran}
\bibliography{references}

\end{document}